\documentclass[draftclsnofoot,10pt,onecolumn,letterpaper]{IEEEtran}

\usepackage{amsmath,amsfonts,amssymb,amsbsy,url,verbatim}
\usepackage{pstricks,psfrag,theorem}
\usepackage{times}
\usepackage[english]{babel}
\usepackage[dvips]{graphicx}
\usepackage{bm}
\usepackage{comment}
\usepackage{stfloats}
\usepackage{cite}
\usepackage{dsfont}
\usepackage{url}
\usepackage{amsmath}
\usepackage{subfigure}

\newtheorem{propos}{Proposition}

\newcommand{\bi}{\begin{itemize}}
\newcommand{\ei}{\end{itemize}}
\newcommand{\ben}{\begin{enumerate}}
\newcommand{\een}{\end{enumerate}}
\newcommand{\bc}{\begin{cases}}
\newcommand{\ec}{\end{cases}}
\newcommand{\bd}{\begin{description}}
\newcommand{\ed}{\end{description}}

\newcommand{\be}{\begin{equation}}
\newcommand{\ee}{\end{equation}}
\newcommand{\bea}{\begin{eqnarray}}
\newcommand{\eea}{\end{eqnarray}}
\newcommand{\me}{\mathrm{e}}

\newcommand{\md}{\mathrm{d}}

\newcommand{\mb}{\mathbf}
\newcommand{\mbb}{\mathbb}

\newcommand{\mr}{\mathrm}

\newcommand{\nn}{\nonumber}
\newcommand{\figw}{\columnwidth}

\graphicspath{%
     {./Figures/}
}

\begin{document}

\title{Performance analysis of an opportunistic relay selection protocol for multi-hop networks \\ {\huge Technical report}
}

\author{Kostas~Stamatiou,~Davide~Chiarotto,~Federico~Librino~and~Michele~Zorzi \\ Department of Information Engineering, University of Padova}

\maketitle

\begin{abstract}
In this technical report, we analyze the performance of an interference-aware opportunistic relay selection protocol for multi-hop line networks which is based 
on the following simple rule: a node always transmits if it has a packet, except when its successive node on the line is transmitting.
We derive analytically the saturation throughput and the end-to-end delay for two and three hop networks, and present simulation results for higher 
numbers of hops. In the case of three hops, we determine the throughput-optimal relay positions.
\end{abstract}


\section{Introduction}
\label{Sec:intro}

Opportunistic routing in multi-hop wireless networks takes advantage of favorable channel conditions in order 
to advance packets over a large distance, thus reducing the end-to-end packet delay. 
In this paper, we consider a  line network consisting of a source, a number of relays and a destination, and
evaluate the performance of an ``interference-aware'' opportunistic relaying protocol; namely, a node always attempts to transmit a packet to the farthest node down the line, except when its successive node is transmitting, in which case it stays silent.
The rationale of the protocol is simple: if two consecutive nodes transmit, it is unlikely that the transmission of the node farther from the destination (FAR) will be successful, due to the strong interference generated from the node closer to the destination (CLOSE); therefore, FAR stays silent in order to avoid interfering with the transmission of CLOSE. 

In the case of a two-hop system, we derive exact expressions for the saturation throughput and the mean end-to-end delay, while
in the case of three hops, we obtain an exact expression for the saturation throughput and an accurate approximation for the 
mean end-to-end delay. The analysis takes fully into account the interaction between the source and relay queues and is based on a generating function approach employed in early work on packet radio networks~\cite{Sidi83}. To the best of our knowledge, these are the first analytical expressions for the delay and throughput of tandem queueing networks with opportunistic routing and a realistic underlying physical layer model that takes into account fading and interference. 

We provide numerical and simulation results for a path-loss and Rayleigh fading channel. In particular, for the case of three hops, 
we determine the relay positions that maximize the saturation throughput. Overall, for typical values of the average link signal-to-noise ratio (SNR), the throughput gain of the considered protocol with respect to an aggressive opportunistic relaying protocol is 10-15$\%$, and even larger with respect to a TDMA protocol. Simulation results for four and five hop systems exhibit similar performance gains. 

Early work on tandem queueing networks~\cite{Sidi87} relied on simplified channel and interference modeling and 
did not consider direct packet transmissions over the distance of multiple hops. Recent work on  
opportunistic routing includes~\cite{Zorzi03,Biswas05,chiarotto10}. Various aspects of line networks have been studied 
in~\cite{Appuswamy10,Vellambi11,Ikki2012}, while~\cite{Chen2012} calculated the end-to-end throughput of dynamic relay selection in a random geometric setting.

\section{System model}
\label{Sec:sysMod}

We consider a slotted-time system, where slot $t \in \mbb{Z}$ is the time interval $[t,t+1)$, and the slot duration, i.e., unity, is equal to the duration of a packet.
The system consists of $N+1$ nodes, i.e., the source, $N-1$ relays and the destination. 
At the end of each slot, a new packet arrives at the end of the source queue with probability $\lambda$ and arrivals are independent across slots (other arrival distributions can also be accommodated by the analysis). The buffer size at the source is infinite. According to the considered protocol, a node transmits its head-of-line packet in slot $t$, if 
its successive node does not transmit in that slot (the last relay always transmits since the destination acts as a sink). 
The packet is kept at the farthest receiver that successfully receives the packet, and is discarded by all others.
If the packet is not successfully received by any receiver, it remains at the head-of-line. 
Finally, we assume that nodes can not transmit and receive simultaneously.
 
For analytical purposes, we make the following assumptions: 
\begin{itemize}
 \item (A1) A packet can cover the distance of at most two hops;
 \item (A2) interference from a transmitter more than two hops away from a receiving node is negligible;
 \item (A3) the buffer size at the relays is unity.
\end{itemize}
(A1) and (A2) are based on the fact that, in terrestial networks, the signal power decreases quickly with distance due to the large path-loss exponent. Therefore, a direct three-hop transmission is highly unlikely for typical SNR values and the interference from far-away transmitters is close to negligible.
These statements are also justified by the simulation results of Section~\ref{Sec:numRes}. Regarding (A3), a relay buffer size larger than unity is unnecessary for $N=2$, since, by virtue of the protocol, the only relay will always transmit if it has a packet, thus it can not receive. For $N\ge3$, a buffer size larger than unity 
could enable a relay to receive a packet in the event that its successive relay transmits. Nevertheless, in Section~\ref{Sec:numRes}, it is demonstrated via simulation that
the protocol performance is insensitive to the relay buffer size for three, four and five hop systems.

Let the numbers $1,\dots,N$ correspond to the source, the first relay,$\dots$, the $(N-1)^{\mr{th}}$ relay. 
In the absence of interference, we denote the probability that a transmission of node $n$ succeeds in covering two hops and one hop as $p_{20,n}$ and $p_{10,n}$, where $p_{10,n}>p_{20,n}$.
(we set $p_{20,N}=0$ since the last relay cannot perform two-hop transmissions). 
The probability that a packet covers at least one hop is $p_{s,n} = p_{10,n} + (1-p_{10,n})p_{20,n}$ (where the subscript
``s'' stands for ``success''). We also define $p_{11,n}, n =1,\dots,N-2$ as the probability of successful reception over a single hop, in the presence of 
interference from transmitter $n+2$, i.e., one hop away from the receiver of $n$. 
Henceforth, we employ the following notation. The complement of $x$ is $\bar{x}$, i.e., $\bar{x} = 1-x$; the derivative of the function $f_x$ is $f'_x$ or $\left. \frac{\md f_x}{\md x} \right|_{x=x_o} = f'_{x_o}$; the double derivative of $f_x$ is $f''_x$ or $\left. \frac{\md^2 f_x}{\md x^2} \right|_{x=x_o} = f''_{x_o}$; and the determinant of matrix $\mb{A}$ is $|\mb{A}|$.

\section{Analysis}
\label{Sec:anal}

Let $Q_1(t),Q_2(t),\dots,Q_{N}(t)$ denote the number of packets at the source, the $1^{\mr{st}}$ relay,$\dots$, the $(N-1)^{\mr{th}}$ relay, respectively.
Since the relay buffer size is unity, $Q_{n}(t) \in \{0,1\}$ for $n = 2,\dots,N$, while, for the source, $Q_1(t) \in \mbb{N}$.
In steady state, the probability generating function (pgf) of the vector $(Q_1(t),\dots,Q_{N}(t))$ is
\begin{equation}
 \label{eq:genFun}
  g_{x_1 x_2 \dots x_{N}} = \mbb{E} \left[ x_1^{Q_1(t)} x_2^{Q_2(t)} \dots x_{N}^{Q_{N}(t)} \right],
\end{equation}
where $(x_1,\dots,x_N) \in [0,1]^N$ is the argument of the pgf. 
Moreover, let $A(t)$ be a Bernoulli random variable with parameter $\lambda$ which represents the arrival (or not)
of a new packet at the source at the end of slot $t$. Then, the pgf of $A(t)$ is $f_{x_1} = \mbb{E}\left[x_1^{A(t)}\right] = \lambda x_1 + \bar{\lambda}$.

The {\em mean end-to-end delay} is calculated as~\cite{Sidi87}
\begin{equation}
 \label{eq:meane2eDelay}
 D = \frac{1}{\lambda} \sum_{n=1}^{N} \left. \frac{\partial g_{x_1\dots x_N}}{\partial x_n} \right|_{x_1=\dots=x_N=1}.
\end{equation}
The {\em saturation throughput} $\tau_s$ is defined as the minimum value of $\lambda$ for which $D$ becomes infinite.

\subsection{Two-hop network ($N=2$)}

In steady state, $g_{x_1x_2}$ satisfies the equation
\begin{align}
 \label{eq:mainFunEq2HopsDef}
 g_{x_1x_2} &= \mbb{E}\left[x_1^{Q_1(t+1)} x_2^{Q_2(t+1)} \right] \nn \\
            &= \mbb{E}\left[x_1^{Q_1(t+1)} x_2^{Q_2(t+1)} \left( \mbb{I}\left(Q_1(t) =0,Q_2(t) =0 \right) + \mbb{I}\left(Q_1(t) >0,Q_2(t) =0 \right) \right. \right. \nn \\
&+ \mbb{I}\left(Q_1(t) =0,Q_2(t) >0 \right) + \left. \left. \mbb{I}\left(Q_1(t) >0,Q_2(t) >0 \right) \right) \right],
\end{align}
where $\mbb{I}(\cdot)$ is the indicator function. From (\ref{eq:mainFunEq2HopsDef}) and (\ref{eq:genFun}),
it follows that $g_{x_1x_2}$ must satisfy the functional equation
\begin{align}
 \label{eq:mainFunEq2HopsInitial}
 g_{x_1x_2} &= f_{x_1} \left[ g_{00} + \left( p_{10}\bar{p}_{20}\frac{x_2}{x_1} + \frac{p_{20}}{x_1} + \bar{p}_{10}\bar{p}_{20} \right) \left(g_{x_10}-g_{00}\right) 
              + \left( \frac{p_{10}}{x_2} + \bar{p}_{10} \right) \left( g_{x_1x_2} - g_{x_10} \right) \right].
\end{align}
The first, second and third terms in the brackets corresponds to the following events: both the source and relay are empty; only the 
source is non-empty, thus a packet advances directly to the destination with probability $p_{20}$, or to the relay with probability $p_{10}\bar{p}_{20}$, or
to neither with probability $\bar{p}_{10}\bar{p}_{20}$; the relay is non-empty (and the source is either empty, or remains silent if it is non-empty), thus the packet
transmission to the destination succeeds with probability $p_{10}$ or fails with probability $\bar{p}_{10}$.

In Proposition~\ref{prop:delay2Hop}, we derive the delay and saturation throughput of a symmetrical two-hop network, i.e., $p_{10,1}=p_{10,2}=p_{10}$.
Since only the source can perform a two-hop transmission, for simplicity we write $p_{20,1}=p_{20}$. 

\begin{propos}
 \label{prop:delay2Hop}
For a symmetrical two-hop network, the mean end-to-end delay is
\begin{equation}
 \label{eq:meanDelay2Hop}
D = \frac{1-\lambda\left(1 - \frac{\bar{p}_{10}\bar{p}_{20}}{p_{10}} \right)}{p_s-\lambda(1+\bar{p}_{20})} + \frac{\bar{p}_{20}}{p_s},
\end{equation}
and the saturation throughput is
\begin{equation}
 \label{eq:satThru2Hop}
 \tau_s = \frac{p_s}{1 + \bar{p}_{20}}.
\end{equation}
\end{propos}
\begin{proof}
Since the size of the relay buffer is unity, from (\ref{eq:genFun}), we have that
\begin{equation}
 \label{eq:genFunLinear2Hop}
 g_{x_1x_2} = g_{x_10} + \left(g_{x_11} - g_{x_10} \right)x_2.
\end{equation}
Substituting in (\ref{eq:mainFunEq2Hops}), we obtain
\begin{align}
 \label{eq:mainFunEq2Hops}
 g_{x_1x_2}f_{x_1}^{-1} &=  g_{00} + \left( p_{10} \bar{p}_{20}x_2 x_1^{-1} + p_{20}x_1^{-1} + \bar{p}_{10}\bar{p}_{20} \right) 
                                    \left(g_{x_10}-g_{00}\right) + \left( p_{10} + \bar{p}_{10}x_2 \right) \left( g_{x_11} - g_{x_10} \right) .
\end{align}
Letting $x_2=\{0,1\}$ in (\ref{eq:mainFunEq2Hops}) yields
\begin{align*}
 g_{x_10} \left(1-f_{x_1}\left(p_{20}x_1^{-1}+\bar{p}_{10}\bar{p}_{20} - p_{10}\right) \right) - g_{x_11}f_{x_1}p_{10} &= g_{00}f_{x_1}\left(1-p_{20}x_1^{-1}-\bar{p}_{10}\bar{p}_{20} \right)  \\
 g_{x_10}f_{x_1}p_s (x_1^{-1}-1 ) + g_{x_11}(f_{x_1}-1) &= g_{00}f_{x_1}p_s(x_1^{-1}-1). 
\end{align*}
Solving this system of equations with respect to $g_{x_10},g_{x_11}$, we obtain
\begin{align}
 g_{x_10} &= g_{00} f_{x_1} \frac{x_1(f_{x_1}-1) + p_{10}p_s(1-x_1)f_{x_1} - (p_{20}+\bar{p}_{10}\bar{p}_{20}x_1)(f_{x_1}-1)}{x_1(f_{x_1}-1)+p_{10}p_s(1-x_1)f_{x_1}^2-
                                \left(p_{20}+ (\bar{p}_{10}\bar{p}_{20}-p_{10})x_1 \right)f_{x_1}(f_{x_1}-1)} 
\label{eq:mainFunEq2HopsProof0}\\
 g_{x_11} &= g_{00}p_s\frac{f_{x_1}}{x_1} \frac{(1-x_1)\left(x_1(1-f_{x_1}) + p_{10}x_1f_{x_1}\right)}{x_1(f_{x_1}-1) + p_{10}p_s(1-x_1)f_{x_1}^2 -
\left(p_{20}+ (\bar{p}_{10}\bar{p}_{20}-p_{10})x_1 \right)f_{x_1}(f_{x_1}-1)}
\label{eq:mainFunEq2HopsProof1}
\end{align}
Letting $x \to 1$ in (\ref{eq:mainFunEq2HopsProof1}) and applying de l'H\^{o}pital's rule results in
\begin{equation*}
\label{eq:probOfEmpty2Hops}
g_{00} = 1-\lambda\frac{1+\bar{p}_{20}}{p_s}.
 \end{equation*}
From (\ref{eq:meane2eDelay}) and (\ref{eq:genFunLinear2Hop}), we have
\begin{equation}
\label{eq:meanDelayInterm2Hop}
D = \frac{1}{\lambda}  \left( \left.\frac{\partial g_{x_11}}{\partial x_1} \right|_{x_1=1} + 1 - g_{10} \right),
\end{equation}
where the first and second terms in the parentheses are the mean queue sizes at the source and relay buffers, respectively (the latter is equal to the 
probability that the buffer is not empty, since the buffer has size unity).
From (\ref{eq:mainFunEq2HopsProof0})-(\ref{eq:mainFunEq2HopsProof1}), with the help of de l'H\^{o}pital's rule, we determine 
$g_{10}$ and $\left.\frac{\md g_{x1}}{\md x} \right|_{x=1}$. After some algebra, we obtain (\ref{eq:meanDelay2Hop}).
Eq.~(\ref{eq:satThru2Hop}) follows from the definition of $\tau_s$.
\end{proof}

As seen in (\ref{eq:satThru2Hop}), the packet arrival rate which saturates the source buffer 
is given by $(p_{10}+p_{20}(1-p_{10}))/(2-p_{20})$. The expression clearly shows the gain of 
opportunistic routing ($p_{20}>0$), with respect to a protocol where two-hop transmissions are not allowed ($p_{20}=0$), in which case $\tau_s = p_{10}/2$.

\subsection{Three-hop network ($N=3$)}

Since $p_{11,n}$ is defined only for $n=1$, for simplicity we write $p_{11,1}=p_{11}$. The system generating function satisfies
\begin{align}
\label{eq:mainFunEqInitial}
 g_{x_1 x_2 x_3} &= f_{x_1} \left[ g_{000} + \frac{a_{x_3}}{x_3} (g_{x_1 x_2 x_3} - g_{x_1 x_2 0} - g_{x_1 0 x_3} + g_{x_1 0 0} + g_{0 0 x_3} - g_{000})  
                    + \frac{b_{x_2 x_3}}{x_2} (g_{x_1 x_2 0} - g_{x_1 0 0}) \right.  \nonumber \\
			     &+ \left. c_{x_1 x_2 x_3} (g_{x_1 0 0} - g_{000}) 
 			     + \frac{d_{x_1 x_2 x_3}}{x_3} (g_{x_1 0 x_3} - g_{x_1 0 0} - g_{0 0 x_3} + g_{000}) \right], 
\end{align}
where
\begin{align}
a_{x_3} &= p_{10,3} + \bar{p}_{10,3}x_3 \nn \\
b_{x_2 x_3} &= \bar{p}_{10,2}\bar{p}_{20,2}x_2 + p_{10,2}\bar{p}_{20,2}x_3 + p_{20,2} \nn \\
c_{x_1 x_2 x_3} &= \bar{p}_{10,1}\bar{p}_{20,1} + p_{10,1}\bar{p}_{20,1}x_2 x_1^{-1} + p_{20,1} x_3 x_1^{-1} \nn \\
d_{x_1 x_2 x_3} &= (p_{11}x_2 x_1^{-1} + \bar{p}_{11})(p_{10,3} + \bar{p}_{10,3}x_3). \label{eq:funcDef3Hop}
\end{align}
The different terms on the right hand side of (\ref{eq:mainFunEqInitial}) can be explained in a fashion similar to (\ref{eq:mainFunEq2HopsInitial}). 
The main difference between the two equations is the presence of the last term in (\ref{eq:mainFunEqInitial}), which captures the event of concurrent transmissions from the source to the first relay
and the second relay to the destination. In Proposition~\ref{prop:thru3Hop}, we derive the end-to-end delay and the saturation throughput of a three-hop network.
\begin{propos}
 \label{prop:thru3Hop}
For a three-hop network, the end-to-end delay is
\begin{equation}
 \label{eq:meanDelay3Hop}
D = \frac{1}{\lambda} \left( 2 + \frac{K''_{111}-2K'_{101}-2K'_{110}}{2K'_{111}} - \frac{K''_1}{2K'_1}  \right) 
\end{equation}
and the saturation throughput is 
\begin{equation}
\label{eq:satThru3HopGeneral}
 \tau_s =  \frac{a_0
\left|
\begin{array}{ccc}
b_{01}     & d_{101}-1   & c_{101}-1  \\
b_{00}     & d_{100}     & c_{100}-1  \\
0          & -d_{111}'    & -c_{111}' 
\end{array}
\right|
}{
\left|
\begin{array}{ccc}
b_{01}   & d_{101}-1   & c_{101}-1  \\
b_{10}-a_0-1  & 0  & c_{110}-a_0-1   \\
b_{00}    &  d_{100}   & c_{100}-1
\end{array}
\right|
},
\end{equation}
where 
\begin{align}
\label{eq:kx11Def}
 K_{x_111} &= \left|
\begin{array}{cccc}
C(d_{x_101}-1) + c_{x_101} -1    & b_{01} 			    & d_{x_101}-f_{x_1}^{-1}          &  c_{x_101} - f_{x_1}^{-1} \\
C(d_{x_110}-a_0) + c_{x_110} -1  & b_{10}-a_0 - f_{x_1}^{-1}        & d_{x_110}-a_0               &  c_{x_110} - a_0  - f_{x_1}^{-1}  \\
C(d_{x_100}-a_0) + c_{x_100} -1  & b_{00}			    & d_{x_100}                   &  c_{x_100} -f_{x_1}^{-1} \\
C(d_{x_111}-1) + c_{x_111} -1    & 0                            & d_{x_111}-1                 &  c_{x_111}-1
\end{array}
\right|
\end{align}
\begin{align}
\label{eq:kx01Def}
K_{x_101} &= \left|
\begin{array}{cccc}
0                            & b_{01}                       & C(d_{x_101}-1) + c_{x_101} -1      & c_{x_101}-d_{x_101}-b_{01}\\
a_0                          & b_{10}-a_0 - f_{x_1}^{-1}        & C(d_{x_110}-a_0) + c_{x_110} -1    & c_{x_110}-d_{x_110}+a_0-b_{10}\\
0                            & b_{00}                       & C(d_{x_100}-a_0) + c_{x_100} -1    & c_{x_100}-d_{x_100}-b_{00}-f_{x_1}^{-1}\\
1-f_{x_1}^{-1}                   & 0                            & C(d_{x_111}-1) + c_{x_111} -1      & c_{x_111}-d_{x_111}
\end{array}
\right|
\end{align}
\begin{align}
\label{eq:kx10Def}
 K_{x_110} &= \left|
\begin{array}{cccc}
0                            & C(d_{x_101}-1) + c_{x_101} -1    & d_{x_101}-f_{x_1}^{-1}          & c_{x_101}-d_{x_101}-b_{01} \\
a_0			     & C(d_{x_110}-a_0) + c_{x_110} -1  & d_{x_110}-a_0               & c_{x_110}-d_{x_110}+a_0-b_{10}\\
0                            & C(d_{x_100}-a_0) + c_{x_100} -1  & d_{x_100}                   & c_{x_100}-d_{x_100}-b_{00}-f_{x_1}^{-1}\\
1-f_{x_1}^{-1}                   & C(d_{x_111}-1) + c_{x_111} -1    & d_{x_111}-1                 & c_{x_111}-d_{x_111}
\end{array}
\right|
\end{align}
\begin{align}
\label{eq:kxDef}
K_{x_1} &= \left|
\begin{array}{cccc}
0                            & b_{01} 			    & d_{x_101}-f_{x_1}^{-1}          &  c_{x_101} - f_{x_1}^{-1} \\
a_0                          & b_{10}-a_0 - f_{x_1}^{-1}        & d_{x_110}-a_0               &  c_{x_110} - a_0  - f_{x_1}^{-1}  \\
0                            & b_{00}			    & d_{x_100}                   &  c_{x_100} -f_{x_1}^{-1} \\
1-f_{x_1}^{-1}                   & 0                            & d_{x_111}-1                 &  c_{x_111}-1
\end{array}
\right|
\end{align}
and all derivatives in (\ref{eq:meanDelay3Hop})-(\ref{eq:satThru3HopGeneral}) are taken with respect to $x_1$. The constant $C$ is defined as
\begin{equation}
\label{eq:CDef}
C \triangleq \frac{g_{001}}{g_{000}}-1.
\end{equation}
\end{propos}
\begin{proof}
Since the size of the relay buffers is unity, we have
\begin{align}
 \label{eq:linearityRelayBuffers}
 g_{x_1 x_2 x_3} &= g_{x_1 x_2 0} + \left(g_{x_1 x_2 1} - g_{x_1 x_2 0} \right)x_3 \nn \\
 g_{x_1 x_2 x_3} &= g_{x_1 0 x_3} + \left(g_{x_1 1 x_3} - g_{x_1 0 x_3} \right)x_2.
\end{align}
Therefore, (\ref{eq:mainFunEqInitial}) becomes
\begin{align}
\label{eq:mainFunEq}
 g_{x_1 x_2 x_3} &= f_{x_1} \left[ g_{000} + a_{x_3} (g_{x_1 x_2 1} - g_{x_1 x_2 0} - g_{x_1 0 1} + g_{x_1 0 0} + g_{0 0 1} - g_{000}) 
+ b_{x_2 x_3} (g_{x_1 1 0} - g_{x_1 0 0}) \right. \nonumber \\
			     &+ \left. c_{x_1 x_2 x_3} (g_{x_1 0 0} - g_{000}) 
 			     + d_{x_1 x_2 x_3} (g_{x_1 0 1} - g_{x_1 0 0} - g_{0 0 1} + g_{000}) \right].
\end{align}
Setting $(x_2,x_3) = \{(0,0),(0,1),(1,0),(1,1)\}$ in (\ref{eq:mainFunEq}) and $x_1=x$ (to make the notation easier), we obtain
\begin{align}
   &b_{01} g_{x10} + (d_{x01}-f_x^{-1}) g_{x01} + (-b_{01} + c_{x01} - d_{x01})g_{x00} = (-C-1 + c_{x01} + Cd_{x01})g_{000} \label{eq:system1} \\
   &a_0g_{x11} + (-a_0 - f_x^{-1} + b_{10}) g_{x10} + (-a_0 + d_{x10}) g_{x01} + (a_0 - b_{10} + c_{x10} - d_{x10}) g_{x00} = \nn \\ 
   &(-1-a_0C + c_{x10} + Cd_{x10})g_{000}  \label{eq:system2} \\   
   & b_{00}g_{x10} + d_{x00}g_{x01} + (-f_x^{-1} - b_{00} + c_{x00} - d_{x00})g_{x00} = (-1-a_0C+c_{x00} +Cd_{x00})g_{000} \label{eq:system3} \\
   & (1-f_x^{-1})g_{x11} + (-1+d_{x11})g_{x01} + (c_{x11}-d_{x11})g_{x00} = (-C-1 +c_{x11} +Cd_{x11})g_{000} \label{eq:system4},   
\end{align}
where the constant $C$ is defined in (\ref{eq:CDef}). Solving (\ref{eq:system1})-(\ref{eq:system4}) over $g_{x11},g_{x10},g_{x01},g_{x00}$, we obtain
\begin{equation}
\label{eq:funsOfK}
 g_{x11}=g_{000}\frac{K_{x11}}{K_x},\ g_{x10}=g_{000}\frac{K_{x10}}{K_x},\ g_{x01}=g_{000}\frac{K_{x01}}{K_x},\ g_{x00}=g_{000}\frac{K_{x00}}{K_x},
\end{equation}
where $K_{x11},K_{x10},K_{x01},K_{x00}$ are defined in (\ref{eq:kxDef}).

By the law of total probability, $g_{111}=1$. Therefore, taking the limit of $g_{x11}$ for $x \to 1$ and
applying de l'H\^{o}pital's rule, we have 
\begin{equation}
 \label{eq:g000FunOfK}
g_{000} = \frac{K_1'}{K_{111}'},
\end{equation}
where
\begin{equation}
\label{eq:derivK1}
K_1' =  a_0\left|
\begin{array}{ccc}
b_{01}     & d_{101}-1   & c_{101}-1  \\
b_{00}     & d_{100}     & c_{100}-1  \\
0          & -d_{111}'    & -c_{111}' 
\end{array}
\right|
- \lambda
\left|
\begin{array}{ccc}
b_{01}   & d_{101}-1   & c_{101}-1  \\
b_{10}-a_0-1  & 0  & c_{110}-a_0-1   \\
b_{00}    &  d_{100}   & c_{100}-1
\end{array}
\right|
\end{equation}
and
\begin{equation}
 \label{eq:derivK111}
 K_{111}' = a_0 
\left|
\begin{array}{cc}
b_{01}  & d_{111}'(c_{101}-1)-c_{111}'(d_{101}-1)  \\
b_{00}  & d_{111}'(c_{100}-1)-c_{111}'d_{100} 
\end{array}
\right| +
a_0 C
\left|
\begin{array}{cc}
b_{01}  & d_{111}'(c_{101}-1)-c_{111}'(d_{101}-1)  \\
b_{10}-a_0-1  & d_{111}'(c_{110}-a_0-1)
\end{array}
\right|
\end{equation}
From the definitions in (\ref{eq:funcDef3Hop}), it is straightforward to show that all the determinants in (\ref{eq:derivK1})-(\ref{eq:derivK111}) are positive.
Moreover, $C \ge 0$, since $g_{001} = \mbb{P}(Q_1(t)=Q_2(t)=0) \ge \mbb{P}(Q_1(t)=Q_2(t)=Q_3(t)=0)=g_{000}$.
The condition of ergodicity of the Markov chain (i.e., finite delay) is $g_{000}>0$~\cite{Sidi87}, from which (\ref{eq:satThru3HopGeneral}) follows.

We now compute the end-to-end delay. Successively applying de l'H\^{o}pital's rule, and recalling (\ref{eq:g000FunOfK}),
the mean queue size at the source is found to be
\begin{equation}
 \label{eq:meanQueueSourceK}
 \left. \frac{\md g_{x11}}{\md x} \right|_{x=1} = \frac{K_{111}''}{2K_{111}'} - \frac{K_1''}{2K_1'}.
\end{equation}
Moreover, recalling (\ref{eq:linearityRelayBuffers}), the mean queue sizes at the first and second relays (or, equivalently, the busy probabilities since
the size of the buffers is unity) are
\begin{align}
 \label{eq:meanQueueRelaysK}
 \left. \frac{\md g_{1x_21}}{\md x_2} \right|_{x_2=1} = 1-g_{101} = 1 - g_{000}\frac{K'_{101}}{K'_1} = 1 - \frac{K'_{101}}{K'_{111}} \nn \\
 \left. \frac{\md g_{11x_3}}{\md x_3} \right|_{x_3=1} = 1-g_{110} = 1 - g_{000}\frac{K'_{110}}{K'_1} = 1 - \frac{K'_{110}}{K'_{111}} \nn \\
\end{align}
From (\ref{eq:meane2eDelay}) and (\ref{eq:meanQueueSourceK})-(\ref{eq:meanQueueRelaysK}), (\ref{eq:meanDelay3Hop}) follows.
\end{proof}

In the particular case of a symmetrical system, i.e., $p_{10,1}=p_{10,2}=p_{10,3}=p_{10}$ and $p_{20,1}=p_{20,2}=p_{20}$, the expression for the saturation throughput
given in (\ref{eq:satThru3HopGeneral}) simplifies considerably. The result is stated in the following proposition.

\begin{propos}
\label{prop:thru3HopSym}
The saturation throughput of a symmetrical three-hop system is 
\begin{equation}
\label{eq:satThru3Hop}
\tau_s = u(p_{10},p_{20},p_{11})v(p_{10},p_{20},p_{11})^{-1}, 
\end{equation}
where
\begin{align*} 
u(p_{10},p_{20},p_{11}) &= p_{10} (p_s^2p_{11}+p_s^2p_{10}\bar{p}_{11}+p_{20}^2p_{11}) \\
v(p_{10},p_{20},p_{11}) &=  p_{10}(1+\bar{p}_{20})(p_sp_{11}+p_{10}^2\bar{p}_{20}\bar{p}_{11}) 
+ (p_{10}+p_{20})(p_sp_{10}\bar{p}_{11}+p_{20}p_{11}). \nn 
\end{align*}
\end{propos}
\begin{proof}
Follows directly from (\ref{eq:satThru3HopGeneral}) by setting  $p_{10,1}=p_{10,2}=p_{10,3}=p_{10}$ and $p_{20,1}=p_{20,2}=p_{20}$.
\end{proof}

Eq.~(\ref{eq:satThru3Hop}) is amenable to interpretation for particular values of the parameters $p_{10},p_{20},p_{11}$.
For example, letting $p_{20}=0$, i.e., not allowing two-hop transmissions, yields
\begin{equation*}
 \tau_s = \frac{p_{10}}{2+\frac{p_{10}\bar{p}_{11}}{p_{11}+ p_{10}\bar{p}_{11}}}.
\end{equation*}
For $p_{11} > 0$, the denominator is $<3$, which reflects the gain with respect to a system where intra-route spatial reuse
is not permitted ($p_{11}=0$ and $\tau_s = p_{10}/3$).

The derivation of a closed-form expression for $D$ from (\ref{eq:meanDelay3Hop}) requires the constant $C$.
We were not able to determine $C$ analytically, but an approximation may be obtained as follows. 
Considering a symmetrical system\footnote{The approximation may be obtained easily for a non-symmetrical system as well.}, and setting $x_1=x_2=x_3=0$ in (\ref{eq:mainFunEq}) gives
\begin{equation*}
 p_{10}g_{001} + p_{20}g_{010} = (-1+\bar{\lambda}^{-1}+p_{10}+p_{20})g_{000}.
\end{equation*}
Letting $g_{001} \approx g_{010}$, we have
\begin{equation}
\label{eq:Capprox}
 C \approx \frac{\lambda}{\bar{\lambda}(p_{10} + p_{20})}.
\end{equation}
Since $g_{001}=\mbb{P}(Q_1(t)=Q_2(t)=0)$ and $g_{010}=\mbb{P}(Q_1(t)=Q_3(t)=0)$, we are approximating the probabilities that
the first two queues are empty, and that the first and third
queues are empty, as equal. Note that (\ref{eq:Capprox}) is proportional to $\lambda$. This is reasonable, since, for $\lambda \to 0$,
$\mbb{P}(Q_1(t)=Q_2(t)=0,Q_3(t)=1) \to 0$ or $g_{001} \to g_{000}$. The accuracy of the approximation is demonstrated with numerical results in the following section.

\section{Numerical results}
\label{Sec:numRes}

We initially present numerical results for symmetrical two and three hop systems (Figs.~\ref{fig:2hop}-\ref{fig:5hop}) and in Fig.~\ref{fig:3hopOptPos}, we examine a non-symmetrical
three-hop system. The considered channel model consists of path-loss $r^{-\alpha}$ at distance $r$, where $\alpha$ is the path-loss exponent, and fading $h$
which is constant within a slot, and spatially and temporally independent. We assume that $h$ is exponentially distributed with mean one (i.e., $\sqrt{h}$ 
is Rayleigh distributed). The (instantaneous) received power is $P_r = P r^{-\alpha}h$, where $P$ is the transmit power, assumed common for all nodes. We define the average received SNR over a single hop as $\gamma = Pr^{-\alpha}/N_0$, where $N_0$ is the thermal noise power. Assuming that a packet is successfully received if the received signal-to-interference-and-noise ratio (SINR) is larger than a threshold $\theta$, the probabilities $p_{10},p_{20},p_{11}$ defined in Section~\ref{Sec:sysMod} are
\begin{equation*}
 p_{10} = \me^{-\theta/\gamma},\ p_{20} = \me^{-2^{\alpha}\theta/\gamma},\ p_{11} = p_{10}/(1+\theta).
\end{equation*}

Apart from the considered ``smart'' opportunistic protocol (S-OPP), described in Section~\ref{Sec:sysMod}, for comparison purposes we consider the following two protocols.

\noindent
{\em Multi-hop (MH):} packets can only be transmitted over a single hop and nodes are divided in groups based on their spatial separation $d=1,\dots,N$ (in hops). 
In each slot, all nodes in a group can transmit simultaneously, and, across slots, a TDMA schedule is followed amongst the groups. If $d =1$, all nodes can transmit in a given slot (full spatial reuse), while, if $d=N$, MH becomes a pure TDMA (round-robin) protocol. 

\noindent 
{\em Regular opportunistic (OPP):} The only difference between OPP and S-OPP is that if a node has a packet in its queue, it transmits, independently of the queue state of the 
successive node.

Note that, in terms of feedback, MH only requires that the transmitter know whether its successive node successfully received the packet. 
In general, OPP and S-OPP require a more refined feedback, since a transmission has multiple potential receivers and all of them have to be informed of the outcome. This can be accomplished
within a separate feedback slot, where, in a round-robin fashion, each node in the network (excluding the source) declares if it successfully received a packet and from which node.
On the other hand, OPP and S-OPP do not require the scheduling of packet transmissions on which MH is based.

For each protocol, we determine via simulation the average delay of the packets that arrive at the destination over a period of $10^6$ slots.
In the simulations, we relax assumptions (A1)-(A3), allowing
for direct transmissions over distances exceeding two hops (if SINR $> \theta$ is satisfied), taking into account interference from all transmitting nodes, and letting the relay buffer size $B_r \ge 1$. The implication of $B_r>1$ is that a relay which has a packet in its buffer at time $t$, may receive a packet in slot $t$ if it is silent. 
Unless otherwise stated, $\alpha=3$, $\gamma = 8$~dB, $\theta=3$~dB and $B_s = B_r = 50$, where $B_s$ denotes the source buffer size.

In Figs.~\ref{fig:2hop}-\ref{fig:3hop}, the delay is plotted vs. $\lambda$ for a two and a three hop system, respectively. Expectedly, S-OPP outperforms OPP and MH (for all possible $d$).
Under little traffic, it is as aggressive as OPP, harnessing good fading conditions to perform direct two-hop transmissions. Under high traffic, it still behaves opportunistically, but avoids causing unnecessary interference, yielding a throughput gain of about 10$\%$ with respect to OPP under saturation. In fact, Fig.~\ref{fig:2hop} depicts nicely how OPP suffers from interference for high traffic, resulting in larger delay than MH for $\lambda>0.33$. Note that the analytical approximation of the delay in Fig.~\ref{fig:3hop} is satisfactory for all $\lambda$.

In Figs.~\ref{fig:4hop}-\ref{fig:5hop}, the simulated delay is plotted vs. $\lambda$ for four and five hop networks. The MH curves are obtained by selecting the delay optimal $d$ for each $\lambda$ (which is $d=1$ for the given system parameter values). For $\gamma=5,10$~dB, the maximum throughput of S-OPP is about 15$\%$ and 10$\%$ larger than OPP, respectively.
For $\gamma=5$~dB in particular, the curves of OPP and MH are overlapping, due to the fact that two-hop transmissions are very rare. This implies that the gain of S-OPP with respect to OPP 
results only from the smart interference management. Another interesting observation is that the performance of S-OPP is insensitive to $B_r$ (seen by the light lines which correspond to $B_r=1$). The reason is that the events where three or more consecutive nodes have packets to transmit are quite rare for $\lambda$ smaller than the saturation throughput; therefore a relay buffer size larger than unity does not result in a notable end-to-end delay benefit.

Closing the paper, we consider the performance of S-OPP in a three-hop system with {\em non-equidistant} relays. In Fig.~\ref{fig:3hopOptPos}, we plot the relay positions 
that maximize $\tau_s$ as a function of $\gamma$, and compare them with the respective ones obtained via simulation of a saturated system.
Note that $\gamma$ in the non-symmetrical case is defined as $\gamma = 3^{\alpha}\gamma_{\mr{tot}}$, where $\gamma_{\mr{tot}}$ is the end-to-end receive SNR.
For normalization purposes, we set the source-destination distance to unity. 
If $r_1,r_2 \in (0,1)$, $r_1<r_2$, denote the distances of the first and second relays from the source, the probabilities defined in Section~\ref{Sec:sysMod} are given by
\begin{align}
 p_{10,1} &= \exp \left( -\frac{(3r_1)^{\alpha}\theta}{\gamma} \right),\ p_{20,1} = \exp \left( {-\frac{(3r_2)^{\alpha}\theta}{\gamma}} \right) \nn \\
 p_{10,2} &= \exp \left( -\frac{(3(r_2-r_1))^{\alpha}\theta}{\gamma}\right),\ p_{20,2} = \exp \left(-\frac{(3(1-r_1))^{\alpha}\theta}{\gamma} \right) \nn \\
 p_{10,3} &= \exp \left( -\frac{(3(1-r_2))^{\alpha}\theta}{\gamma} \right) \nn \\
 p_{11} &= \frac{p_{10,1}}{1 + \theta \left(\frac{r_1}{r_2-r_1} \right)^{\alpha}} \label{eq:probs3HopGeneral}
\end{align}
and $\tau_s$ is obtained by  
\begin{equation}
\label{eq:satThru3HopCompicated}
 \tau_s =  \frac{p_{10,3}
\left|
\begin{array}{ccc}
p_{10,2}\bar{p}_{20,2}+p_{20,2}     & -p_{11}   & \bar{p}_{10,1}\bar{p}_{20,1}+p_{20,1}-1  \\
p_{20,2}     & \bar{p}_{11}p_{10,3}     & \bar{p}_{10,1}\bar{p}_{20,1}-1  \\
0          & p_{11}    & p_{10,1}\bar{p}_{20,1}+p_{20,1} 
\end{array}
\right|
}{
\left|
\begin{array}{ccc}
p_{10,2}\bar{p}_{20,2}+p_{20,2}     & -p_{11}   & \bar{p}_{10,1}\bar{p}_{20,1}+p_{20,1}-1 \\
\bar{p}_{10,2}\bar{p}_{20,2}+p_{20,2}-p_{10,3}-1  & 0  & \bar{p}_{20,1}-p_{10,3}-1   \\
p_{20,2}    &  \bar{p}_{11}p_{10,3}   & \bar{p}_{10,1}\bar{p}_{20,1}-1
\end{array}
\right|
}. 
\end{equation}
$\tau_s$ is evaluated as a function of $\gamma$ by substituting (\ref{eq:probs3HopGeneral}) in (\ref{eq:satThru3HopCompicated}).
The discrepancy of the theoretical and simulated curves for $\gamma >10$~dB observed in Fig.~\ref{fig:3hopOptPos} is due to the fact that, in the simulated system, the
destination can be reached directly from the source with positive probability, which is not taken into account in the analysis. 
Focusing on the more realistic SNR range $6-10$~dB, the main conclusion drawn from 
Fig.~\ref{fig:3hopOptPos} is that, under normal S-OPP operation, it is advantageous to move the first relay slightly closer to the destination than 0.33. This position achieves the best tradeoff between reducing interference from the second relay and advancement towards the destination. This can be confirmed by the curves obtained when either two-hop transmissions or intra-route reuse are forbidden.

\bibliographystyle{IEEEtran}
\bibliography{IEEEabrv,References}

\begin{thebibliography}{1}
\providecommand{\url}[1]{#1}
\csname url@samestyle\endcsname
\providecommand{\newblock}{\relax}
\providecommand{\bibinfo}[2]{#2}
\providecommand{\BIBentrySTDinterwordspacing}{\spaceskip=0pt\relax}
\providecommand{\BIBentryALTinterwordstretchfactor}{4}
\providecommand{\BIBentryALTinterwordspacing}{\spaceskip=\fontdimen2\font plus
\BIBentryALTinterwordstretchfactor\fontdimen3\font minus
  \fontdimen4\font\relax}
\providecommand{\BIBforeignlanguage}[2]{{%
\expandafter\ifx\csname l@#1\endcsname\relax
\typeout{** WARNING: IEEEtran.bst: No hyphenation pattern has been}%
\typeout{** loaded for the language `#1'. Using the pattern for}%
\typeout{** the default language instead.}%
\else
\language=\csname l@#1\endcsname
\fi
#2}}
\providecommand{\BIBdecl}{\relax}
\BIBdecl

\bibitem{Sidi83}
M.~Sidi and A.~Segall, ``Two interfering queues in packet-radio networks,''
  \emph{{IEEE} Trans. Commun.}, vol.~31, pp. 123--129, Jan. 1983.

\bibitem{Sidi87}
M.~Sidi, ``Tandem packet-radio queueing systems,'' \emph{{IEEE} Trans.
  Commun.}, vol.~35, pp. 246--248, Feb. 1987.

\bibitem{Zorzi03}
M.~Zorzi and R.~Rao, ``Geographic random forwarding ({GeRaF}) for ad hoc and
  sensor networks: multihop performance,'' \emph{{IEEE} Trans. Mobile Comput.},
  vol.~2, pp. 337--348, Oct. 2003.

\bibitem{Biswas05}
S.~Biswas and R.~Morris, ``{ExOR}: Opportunistic multi-hop routing for wireless
  networks,'' in \emph{ACM SIGCOMM}, Aug. 2005, pp. 133--144.

\bibitem{chiarotto10}
D.~Chiarotto, O.~Simeone, and M.~Zorzi, ``Throughput and energy efficiency of
  opportunistic routing with {type-I HARQ} in linear multihop networks,'' in
  \emph{IEEE GLOBECOM}, Dec. 2010, pp. 1--6.

\bibitem{Appuswamy10}
R.~Appuswamy, E.~Atsan, C.~Fragouli, and M.~Franceschetti, ``On relay placement
  for deterministic line networks,'' in \emph{IEEE Wireless Network Coding
  Conference (WiNC)}, Jun. 2010, pp. 1--9.

\bibitem{Vellambi11}
B.~Vellambi, N.~Torabkhani, and F.~Fekri, ``Throughput and latency in
  finite-buffer line networks,'' \emph{{IEEE} Trans. Inf. Theory}, vol.~57, pp.
  3622--3643, Jun. 2011.

\bibitem{Ikki2012}
S.~Ikki and S.~Aissa, ``Multihop wireless relaying systems in the presence of
  cochannel interferences: performance analysis and design optimization,''
  \emph{{IEEE} Trans. Veh. Technol.}, vol.~61, pp. 566 --573, Feb. 2012.

\bibitem{Chen2012}
Y.~Chen and J.~Andrews, ``An upper bound on multihop transmission capacity with
  dynamic routing selection,'' \emph{{IEEE} Trans. Inf. Theory}, vol.~58, pp.
  3751--3765, Jun. 2012.

\end{thebibliography}

\newpage

\begin{figure}
	\centering
	\includegraphics[width = 0.9\figw]{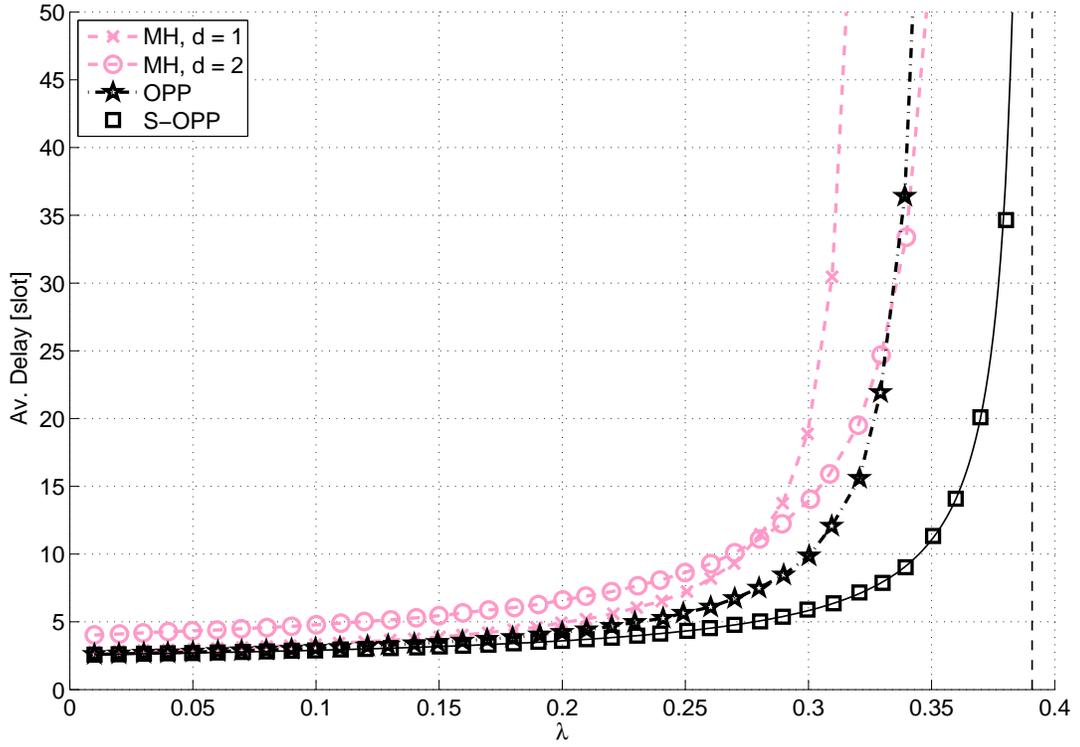}
	\caption{End-to-end delay versus $\lambda$ for a two-hop system. The solid line corresponds to $D$ (\ref{eq:meanDelay2Hop}) and the dotted vertical line to $\tau_s$ (\ref{eq:satThru2Hop}). Note that the simulation markers for S-OPP lie exactly on the theoretical curve. ($\alpha=3$, $\gamma = 8$~dB, $\theta=3$~dB, $B_s = B_r = 50$)}
	\label{fig:2hop}
\end{figure}

\begin{figure}
	\centering
	\includegraphics[width = 0.9\figw]{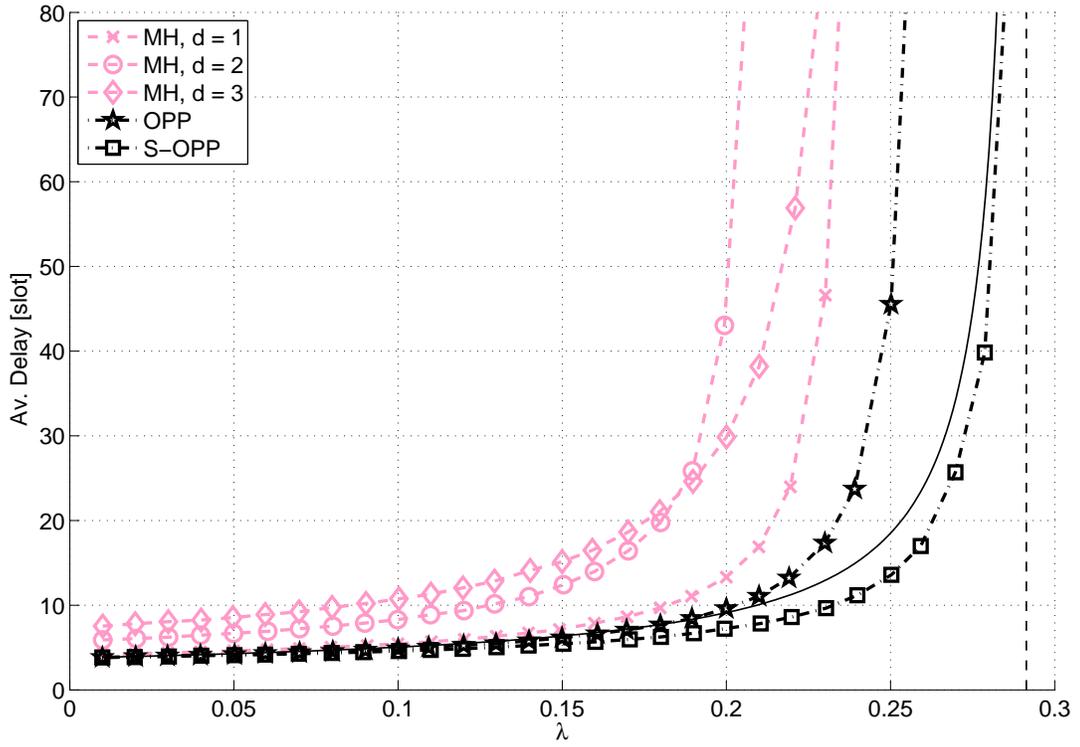}
	\caption{End-to-end delay versus $\lambda$ for a three-hop system. The solid line corresponds to the analytical approximation of $D$ and the dotted vertical line to $\tau_s$ (\ref{eq:satThru3Hop}). ($\alpha=3$, $\gamma = 8$~dB, $\theta=3$~dB, $B_s = B_r = 50$)}
	\label{fig:3hop}
\end{figure}

\begin{figure}
	\centering
	\includegraphics[width = 0.9\figw]{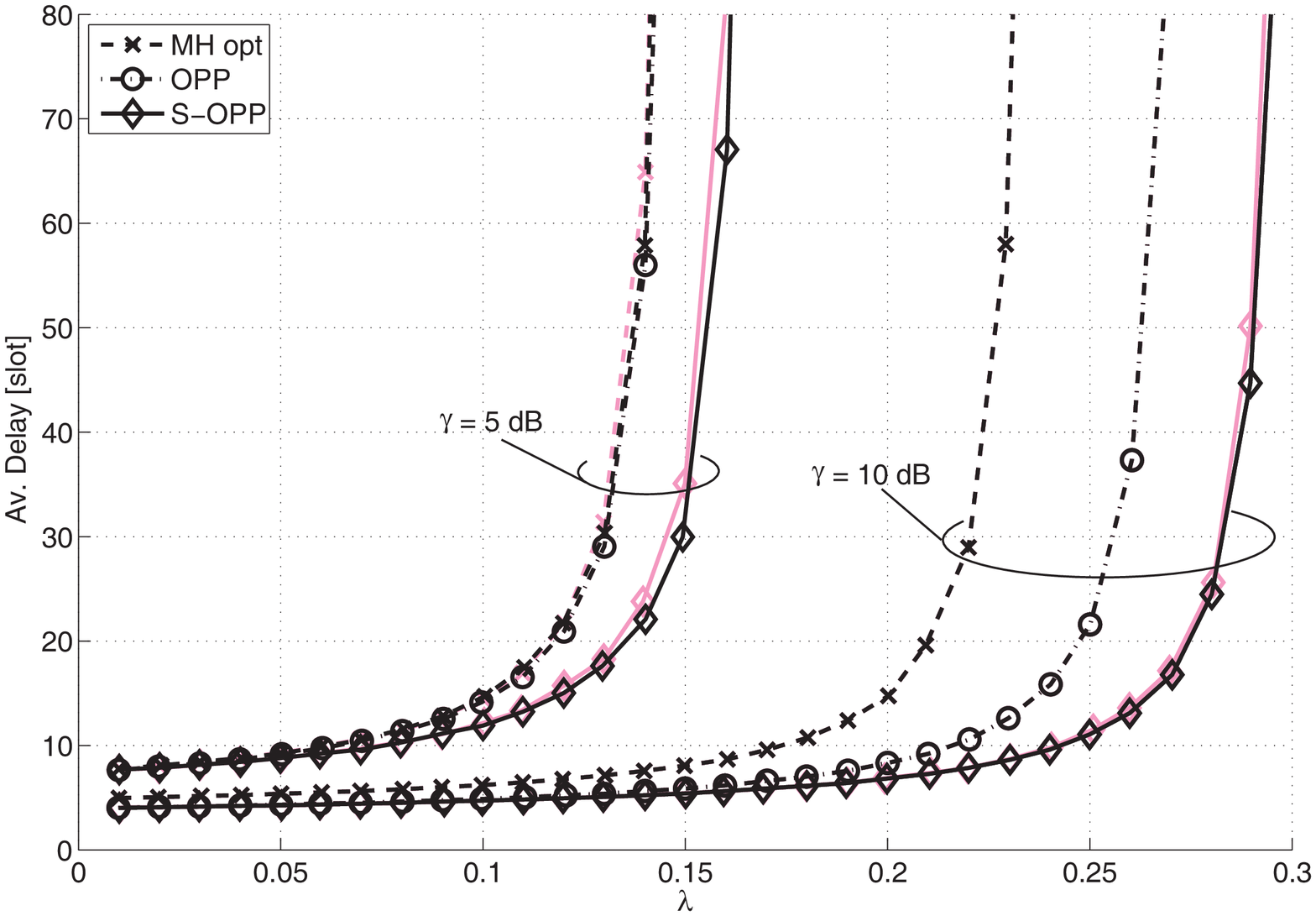}
	\caption{End-to-end delay versus $\lambda$ for a four-hop system and $\gamma=5,10$~dB. The light lines correspond to $B_r=1$ and the dark lines to $B_r=50$. 
                ($\alpha=3$, $\theta=3$~dB, $B_s=50$)}
	\label{fig:4hop}
\end{figure}

\begin{figure}
	\centering
	\includegraphics[width = 0.9\figw]{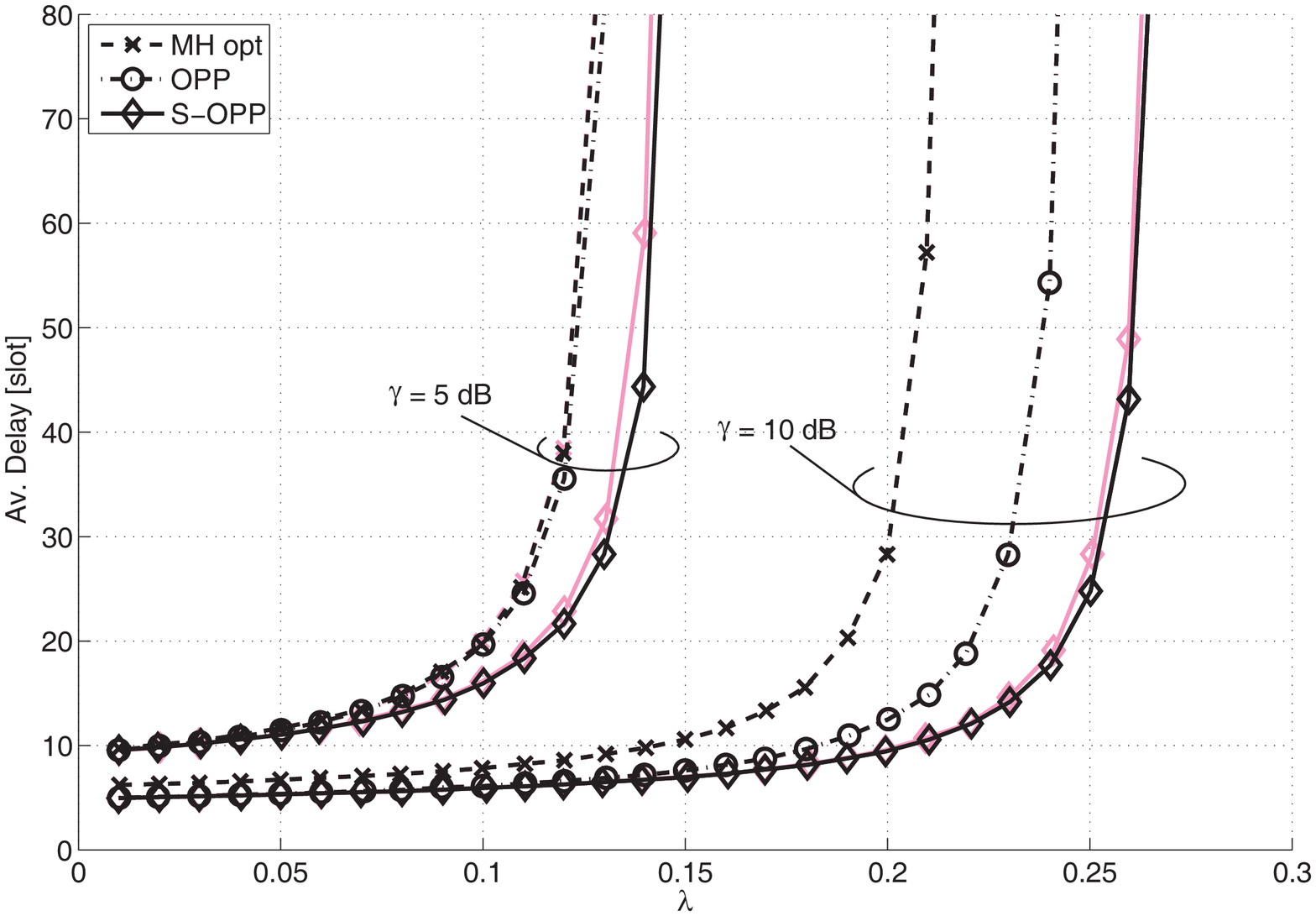}
	\caption{End-to-end delay versus $\lambda$ for a five-hop system and $\gamma=5,10$~dB. The light lines correspond to $B_r=1$ and the dark lines to $B_r=50$.
                ($\alpha=3$, $\theta=3$~dB, $B_s=50$)}
	\label{fig:5hop}
\end{figure}

\begin{figure}
	\centering
	\psfrag{p_{20}=0}{No two-hop}
	\includegraphics[width = 0.9\figw]{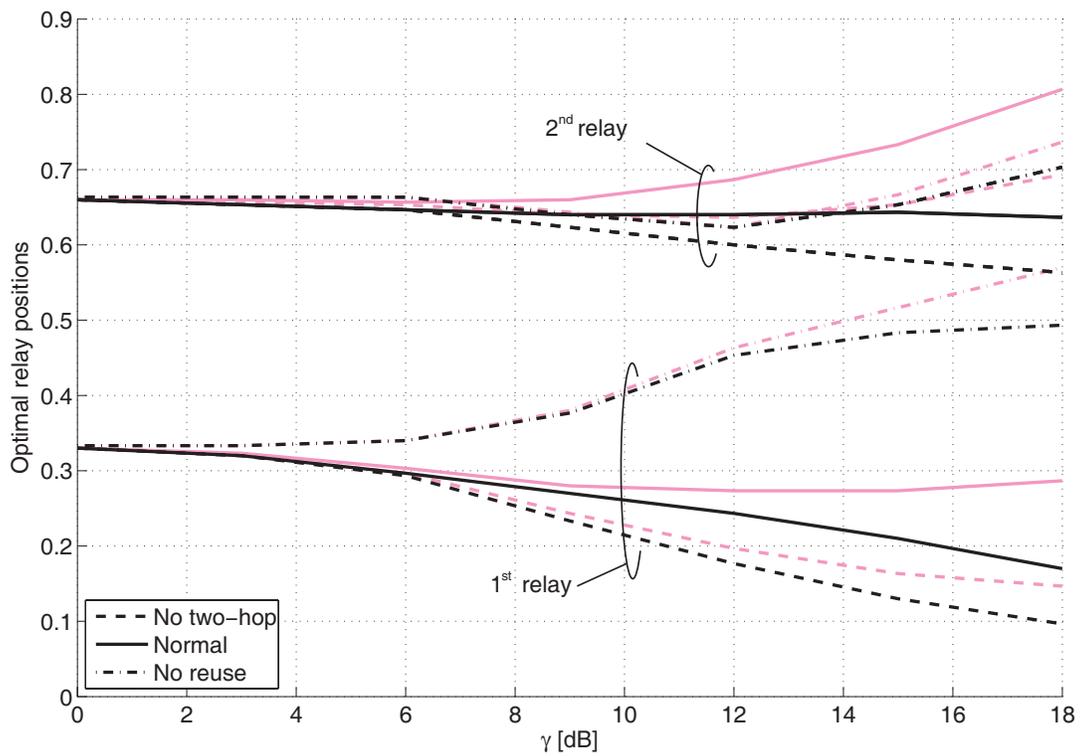}
	\caption{Throughput-optimal relay positions (normalized to unity) for a three-hop S-OPP system vs. equivalent link-SNR $\gamma$. 
         The dark lines are obtained by maximizing the theoretical saturation throughput and the light lines by simulation.
         ($\alpha=3$, $\theta=3$~dB, $B_s=B_r=50$)}
	\label{fig:3hopOptPos}
\end{figure}

\end{document}